\documentclass[11pt]{article}

\usepackage[utf8]{inputenc}
\usepackage[in]{fullpage}
\usepackage{amsfonts, amsmath, amssymb}
\usepackage{multirow}
\usepackage{arydshln} 
\usepackage{color}
\usepackage[dvipsnames]{xcolor}
\usepackage{braket}
\usepackage{graphicx}
\usepackage{booktabs}

\usepackage{todonotes}

\usepackage{hyperref}
\hypersetup{
    colorlinks=true,
    linkcolor=blue,
    filecolor=magenta,      
    urlcolor=cyan,
    citecolor=blue,
    pdftitle={Spectral Centrality of Directed Hypergraphs},
    pdfpagemode=FullScreen,
    }

\usepackage[sorting=none,
            style=numeric,
            doi=false,
            isbn=false,
            url=false,
            eprint=false]{biblatex} 
\definecolor{darkgreen}{rgb}{0.09, 0.45, 0.27}
\definecolor{darkmagenta}{rgb}{0.55, 0.0, 0.55}

\usepackage{amsthm}
\newtheorem{theorem}{Theorem}[section]

\newtheorem{definition}[theorem]{Definition}

\newtheorem{proposition}[theorem]{Proposition}

\newcommand{\cT}{\mathcal{T}}
\newcommand{\cZ}{\mathcal{Z}}
\newcommand{\cH}{\mathcal{H}}
\newcommand{\cB}{\mathcal{B}}
\newcommand{\cF}{\mathcal{F}}
\newcommand{\R}{\mathbb{R}}
\newcommand{\bc}{\mathbf{c}}

\usepackage{bm}
\usepackage{dutchcal}   

\usepackage{wrapfig}
\usepackage{tikz}
\tikzstyle{vertex} = [fill,shape=circle,node distance=80pt]
\tikzstyle{edge} = [fill,opacity=.5,fill opacity=.5,line cap=round, line join=round, line width=50pt]
\tikzstyle{elabel} =  [fill,shape=circle,node distance=30pt]

\pgfdeclarelayer{background}
\pgfsetlayers{background,main}

\title{\bf
\rule{\linewidth}{1pt}
Beyond directed hypergraphs: heterogeneous hypergraphs and spectral centralities 
\rule{\linewidth}{1pt}
} 

\makeatletter
\renewcommand\@date{{%
  \vspace{-2\baselineskip}%
  \large\centering
  \begin{tabular}{@{}c@{}}
    Gonzalo Contreras-Aso\footnote{Corresponding author: gonzalo.contreras@urjc.es} \textsuperscript{$,\,\sharp$,$\flat$}
  \end{tabular}%
  \quad and \quad
  \begin{tabular}{@{}c@{}}
    Regino Criado\textsuperscript{$\sharp$,$\flat$,$\natural$}
  \end{tabular}%
  \quad and\quad
  \begin{tabular}{@{}c@{}}
    Miguel Romance\textsuperscript{$\sharp$,$\flat$,$\natural$}
  \end{tabular}

  \bigskip

{\it \normalsize
  \textsuperscript{$\sharp$}Departamento de Matem\'atica Aplicada, Ciencia e Ingenier\'{\i}a de los Materiales y Tecnolog\'{\i}a Electr\'onica, Universidad Rey Juan Carlos, 28933 M\'ostoles (Madrid), Spain\par
  \textsuperscript{$\flat$}Laboratory of Mathematical Computation on Complex Networks and their Applications, Universidad Rey Juan Carlos, 28933 M\'ostoles (Madrid), Spain\par
  \textsuperscript{$\natural$}Data, Complex Networks and Cybersecurity Research Institute, Universidad Rey Juan Carlos, 28028 Madrid, Spain
}

  \bigskip

  \today
}}
\makeatother

\begin{document}

\maketitle

\begin{abstract}
    The study of hypergraphs has received a lot of attention over the past few years, however up until recently there has been no interest in systems where higher order interactions are not undirected. In this article we introduce the notion of heterogeneous hypergraphs from an algebraic point of view, which have traditional directed hypergraphs as a particular case. We furthermore analytically study the spectral centralities associated to some types of heterogeneous hypergraphs, extending previously defined eigenvector-like centrality measures to this new realm. We supplement the analytical arguments with some numerical comparisons of pairwise and higher order rankings, and we construct directed higher order networks from real data, which we then use for discussion and analysis.
\end{abstract}


\newpage

\section{Introduction} \label{sec:introduction}
    In recent years, the study of complex systems has transcended binary relationships and has embraced the multifaceted nature of interactions prevalent in real-world scenarios \cite{battiston2021physics}. Traditional graph structures, while suitable for representing pairwise relationships in networks, encounter limitations when grappling with interactions involving multiple entities simultaneously. Hypergraphs \cite{berge1984hypergraphs}, an extension of conventional graphs, provide a more expressive framework by accommodating hyperedges that link multiple vertices \cite{battiston2020networks,boccaletti2023structure}. This expanded structure allows for the representation of richer, multi-entity relationships, paving the way for a more comprehensive understanding of complex systems across various disciplines. 
    
    The concept of centrality, a cornerstone in network analysis, encapsulates various metrics that quantify the relative significance of nodes within a network \cite{boccaletti2006complex, estrada2012structure, newman2010networks}. A subset of these measures (the so-called spectral centralities \cite{vigna2016spectral}) are of paramount importance both for the possibility of analytical treatment that they provide and for the reduced computational cost that they usually entail when compared to other measures such as the betweenness centrality. However, their applicability and adaptations to hypergraphs pose intriguing challenges and opportunities for unveiling essential features of complex systems' structures and dynamics \cite{benson2019three}. 
    
    The task of translating traditional network theoretical concepts such as that of centrality measures, as well as other tools and paradigms to the undirected hypergraph realm is already a challenging endeavor, so much so that little to no attention has been paid to the case of directed hypergraphs. These structures have been known for some decades now \cite{gallo1993directed}, although the bulk of works with them have belonged to purely mathematical or computer science literature \cite{xie2016spectral,ausiello2017directed}. Nevertheless, it is natural to consider them as the next step for the network science community. 
    
    In the case of standard (pairwise) networks the meaning of directed interactions is clear, but when we consider interactions pertaining to multiple individuals the notion of directedness becomes fuzzy, to say the least. The traditional approach has made a somewhat sensible choice by defining a directed hyperedge as a hyperedge where some of its nodes are the ``input'' and some are the ``output'' \cite{gallo1993directed, gallo2022synchronization}. While this definition is very useful in terms of enabling a vast collection of results, it constrains the space of possible \textit{non-undirected} interactions. For this reason, we find that the term ``heterogeneous hyperedge'' is perhaps a more adequate label for non-undirected hyperedges, including directed hyperedges as a specific type of them. 
    
    The purpose of this paper is twofold: on the one hand we intend to raise awareness about the existence and interest of different kinds of heterogeneous interactions between several individuals, discussing some of the perhaps more interesting examples. On the other hand, we want to analyze how to extend the definitions of spectral centralities already existing for undirected hypergraphs to some of these heterogeneous hyperedges, and we will provide numerical examples of then applied to both real and synthetic hypergraphs. 
    
    This article is structured as follows: in Section \ref{sec:preliminaries} we will introduce the mathematical preliminaries to deal with undirected hypergraphs and their spectral centralities. In Section \ref{sec:heterogeneous} we delve into the problem of heterogeneity in the hyperedges, conceptually and technically. As we will see, the traditional view of directed hyperedges as set-like objects is rather limit, instead we will advocate for its definition in terms of components of the adjacency tensors/hypermatrices. In Section \ref{sec:directed} the concept of spectral centralities is extended to some simple cases of heterogeneity, including the well-known directed hypergraph case, supplemented by numerical simulations. In Section \ref{sec:conclusions} we conclude our study with a summary. 

\newpage

\section{Preliminaries and notation} \label{sec:preliminaries}

Let us start by defining a hypergraph based on the original definition, although \cite{berge1984hypergraphs} explicitly introducing the possibility of having directed hyperedges.
\begin{definition}[Hypergraph]
    A hypergraph $H = (V, E)$ consists of a vertex set $V$ and a hyperedge set $E$, where $V$ is a non-empty set of nodes or vertices and $E = \{e_1, e_2, \dots, e_m\}$ is a collection of hyperedges. Each hyperedge $e_i \subseteq V$ is an ordered subset of vertices, possibly containing more than two vertices.
\end{definition}

Hypergraphs generalize traditional graphs (where each hyperedge contains only two nodes) by allowing hyperedges to connect multiple vertices, offering a more expressive way to model relationships and interactions involving multiple elements simultaneously \cite{battiston2020networks,boccaletti2023structure}.

Let us stress the word ``ordered'': the bulk of the complex systems literature using hypergraphs has studied the undirected case, which is the simplest and assumes hyperedges to be unordered subsets. One of the goals of this manuscript is moving beyond this simple assumption. 

We will be analyzing the $\cH$-eigenvector centrality (HEC) and the $\cZ$-eigenvector centrality (ZEC) of a hypergraph $H=(V,E)$, first defined in \cite{benson2019three}. This is based on the spectral theory of tensors \cite{qi2017tensor}, in particular those associated to strongly connected, $m$-uniform hypergraphs. The discussion on strong connectivity is reserved for the coming section, while uniformity is defined as follows.

{

\begin{definition}[$m$-uniform hypergraph]
    Let $H=(V,E)$ be a hypergraph. $H$ is $m$-uniform if all of its hyperedges are of size $m$, i.e. $|e|=m,\,\forall e\in E$.
\end{definition}

This is a very strong restriction, however it provides us with strong analytical tools, as we now discuss.

\subsection{Algebraic properties}  \label{subsec:prelim-algebra}

Let us start by defining what is the adjacency tensor\footnote{The word ``tensor'' usually refers to mathematical objects satisfying specific transformation properties. Here we are instead making reference to multidimensional arrays (or hypermatrices), which we refer to as tensors for simplicity.} of a hypergraph as well as its strong connectivity.

\begin{definition}[Adjacency tensor of a uniform hypergraph] \label{def:adjacency-tensor}
Let $H=(V,E)$ be an $m$-uniform, unweighted hypergraph. Let $|V|=N$. Its associated adjacency tensor $\cT=(T_{i_1 ... i_m}) = \R^{[m,N]}$ is given by
\begin{equation}
    T_{i_1 ... i_m} =
    \begin{cases}
        1 & \text{if } (i_1, \dots, i_m) \in E \\
        0 & \text{otherwise.}
    \end{cases}
\end{equation}
\end{definition}

Here the notation $\R^{[m,N]}$ makes it explicit the fact that these objects are really just multidimensional arrays of size $\underbrace{N\times N\times \dots \times N}_{m}$.

The generalization to weighted hypergraphs is straightforward. Notice that we have not made any specific assumptions on the nature of the hyperedges, in particular we have not stated that they are undirected. We can already foretell that the order in $(i_1, \dots, i_m)$ will be of relevance when we discuss directionality.

\begin{definition}[Strongly connected hypergraph \cite{benson2019three}]\label{def:strong-conn}
Let $H=(V,E)$ be a $m$-uniform hypergraph with associated tensor $\cT\in\R^{[m,N]}$. $H$ is said to be strongly connected if the graph $G^M$ induced by the matrix $M=(M_{ij}) = \sum_{j_3 ... j_{m}} T_{i j j_3 ... j_m}$ is strongly connected.
\end{definition}

Clearly, the strong connectivity is completely reliant on the structure of the adjacency tensor. This will become very important in the case of directed hypergraphs.  {Note that, even when a hyperedge is not undirected, this definition still makes sense, as it yields a weight for the relation between two nodes, summing others involved in the larger interaction.}

Let us now define the so-called ``tensor apply'' operation \cite{benson2019computing}, which is nothing but the contraction of tensor $\mathcal{T}\in\R^{m,N}$ with vector $\bc \in \R^N$ producing yet another vector:
\begin{equation}\label{def:tensor-apply}
    \bm{x} = T \bc^{m-1} \Longleftrightarrow  x_{i_1} = \sum_{i_2...,i_m=1}^n T_{i_1 i_2 ... i_m} \bc_{i_2} ... \bc_{i_m}.
\end{equation}

This operation is the basis of the tensor eigenvector problem we will now discuss.

\subsection{The $\cH$-eigenvector centrality} \label{subsec:eiglike-centralities}

In \cite{benson2019three}, the author builds upon the theory of tensor eigenvalues \cite{qi2017tensor}, to define three spectral centrality measures for uniform hypergraphs: the clique-eigenvector centrality (CEC), the $\cZ$-eigenvector centrality (ZEC) and the $\cH$-eigenvector centrality (HEC). 

Here we will only review the latter, for several reasons. First, the CEC is equivalent to the standard eigenvector centrality of the projected hypergraph (where every hyperedge is substituted by a clique among the participant nodes), therefore missing any nonlinear information contained in the higher-order structure. Second, the ZEC is a tensor-based measure which yields a fixed-norm vector (i.e. it is not re-scalable), has no uniqueness guarantees and is computationally problematic \cite{qi2017tensor, benson2019three}. It is also not compliant with a uniformization procedure as in \cite{contrerasaso2023uplifting}. The HEC suffers none of these problems, as we will see. What's more, all ideas discussed in this manuscript also apply to the CEC and ZEC cases without extra effort.

Let's start by considering the simplest non-trivial case, that of a tensor $\cT=(T_{ijk})\in \R^{3,N}$. The $\cH$-eigenproblem for this tensor can be stated as
\begin{equation}
    \quad \lambda \bc^{[2]} = \cT \bc^2  \quad \Rightarrow \quad \lambda c_i^2 = \sum_{j,k} T_{ijk} c_j c_k.
\end{equation}
with the same ``tensor apply'' operation as before, and with the notation $\bc^{[m-1]}$ indicating the Hadamard (or componentwise) power of vector $\bc$.

This eigenproblem can be cast in a network science context as a centrality measure \cite{benson2019three}, bearing in mind that it should be applied to the transposed version of the hypergraph's associated tensor,{  $\cT^t$ (whose meaning and definition will be discussed in the next Section).}

\begin{definition}[HEC of a hypergraph \cite{benson2019three}] \label{def:HEC-centrality}
Let $H$ be the a $m$-uniform hypergraph, with an associated tensor $\mathcal{T}=(T_{i_1\dots i_m}) \in \R^{[m,N]}$. Its $\cH$-eigenvector centrality (HEC) is defined as the Perron-like (unique, positive) $\cH$-eigenvector $\bf{c}$ of $\cT^t$, i.e.
\begin{equation}
    \lambda \mathbf{c}^{[m-1]} =  (\mathcal{T})^t \mathbf{c} \dots \mathbf{c} \quad \Rightarrow \quad \lambda c_{i_1}^m = \sum_{i_2\dots i_m=1} (T_{i_1 \dots i_m})^t\, c_{i_2} \dots c_{i_k}, \quad \mathbf{c} > 0, \; \lambda = r(\cT^t),
\end{equation}
{where $r(\cT^t)$ denotes the spectral radius of $\cT^t$.}
\end{definition}


\section{Heterogeneous hypergraph zoology} \label{sec:heterogeneous}

In this Section we want to begin our quest to generalize hypergraphs, to include some notion of directedness. As we advanced in the introduction, so far in the literature the notion of directed hypergraphs has been restricted to a special type of higher order interaction, especially amenable to a set-theoretic description \cite{gallo1993directed}. Our first task is, therefore, breaking free from this constraint and understanding how we can build directed interactions from other points of view. What's more, even within the set theoretic framework we will see that it is easy to find new types of structures. The wide variety of hypergraph types requires a new adjective to describe them all, beyond traditionally directed hypergraphs, which is why we chose to call them heterogeneous hypergraphs. 

Once this issue has been dealt with, we will give some examples of where this description could be useful, and we will end this section with a technical interlude regarding an operation which is especially relevant in directed graphs, the transposition of the adjacency matrix, and how to extend it to higher order situations.

\subsection{Algebra, topology and set theory} \label{subsec:directed-algebra-sets}

Let \textcolor{red}{us} start from the very beginning: upon opening any textbook or scientific publication on graph theory, one is most likely going to find the definition of a graph to be $G=(V,E)$, where $V$ is the set of nodes and $E\subseteq V \times V$ is the set of edges \cite{estrada2012structure, newman2010networks}. The distinction between directed and and directed edges lies, in this context, on whether the edges $e\in E$ are ordered or unordered pairs of nodes. Notice that this is a set theoretic point of view: a graph can also be understood from a topological point of view (nodes and connections between them) or from an algebraic point of view (e.g. adjacency matrix), see Figure \ref{fig:three-perspectives}.

These different approaches are so important that they even have names of their own (topological graph theory, algebraic graph theory), as they make use of different techniques, tools and results. In standard graph theory these three perspectives complement each other, for standard graphs (regardless of directionality) are equivalently described in either.

\begin{figure}[h!]
    \centering
    \includegraphics[scale=0.7]{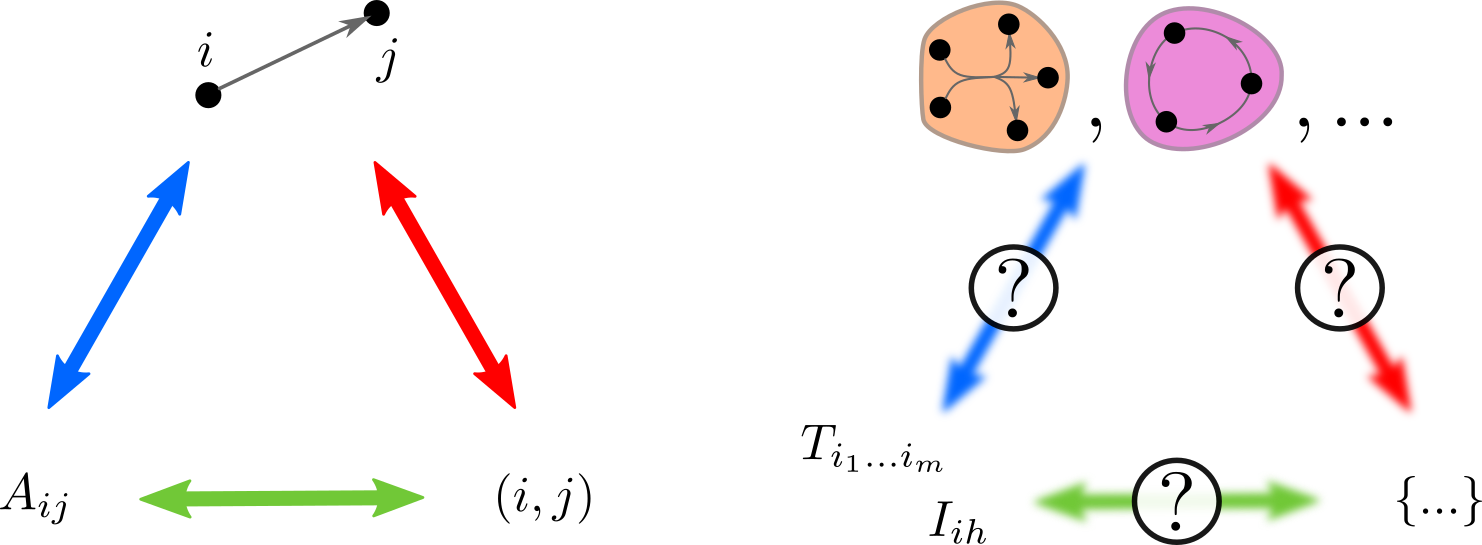}
    \caption{Three different points of view: topological (up), algebraic (down-left) and set-theoretic (down-right). In standard directed graphs (left triangle) the translation between the three perspectives is clear. In heterogeneous hypergraphs there is a vast landscape of possibilities, which makes the dictionary between any two perspectives unclear at best, otherwise impossible.
    }
    \label{fig:three-perspectives}
\end{figure}

\medskip

Starting from the set-theoretic viewpoint, the simplest idea is allowing each hyperedge to have multiple ``input'' nodes, and multiple ``output'' nodes. In \cite{gallo1993directed} this perspective is fleshed out: a \textit{directed} hyperedge\footnote{It should be noted that in some works in the literature \cite{jost2019hypergraph, mulas2021spectral} these hyperedges are called \textit{oriented} instead of \textit{directed}.} is a tuple $E_k=(H(E_k), T(E_k))$ where $H(E_k) \subseteq V$ is the ``head'' of the hyperedge, while $T(E_k) \subseteq V$ is the ``tail'' of the hyperedge, and such that\footnote{According to \cite{gallo1993directed}, either the head or the tail set could be empty although we will assume from here on that they aren't. } $H(E_k) \cap T(E_k) = \emptyset$. We will elaborate on this type of hypergraph structures on \ref{subsec:acyclic-BF}. See Figure \ref{fig:set-directed-hypergraph} for an example. 


This point of view is advantageous if the application in mind can leverage the set-theoretic nature of the hypergraph, or if the directed incidence matrix $I_{ij}$ of the hypergraph can be involved. However, if that is not the case, then this description of hypergraphs might not be suitable for our problem. 

Furthermore, even when within set-theoretic descriptions of higher order interactions, there are plenty of situations that are left out by the input-output paradigm: one could think of a higher order interaction where there are some input nodes, some intermediate nodes, and some output notes (see edge $E_4$ in Figure \ref{fig:set-directed-hypergraph} for an example). In that sense we not only generalize the amount of nodes in the end points of an edge, but we also generalize the amount of stages within a single interaction. This may be relevant, for instance, to represent complicated processes such as computer programs with different functions and modules with intermediate stages.

\begin{figure}[h!]
    \centering
    \includegraphics[width=0.8\textwidth]{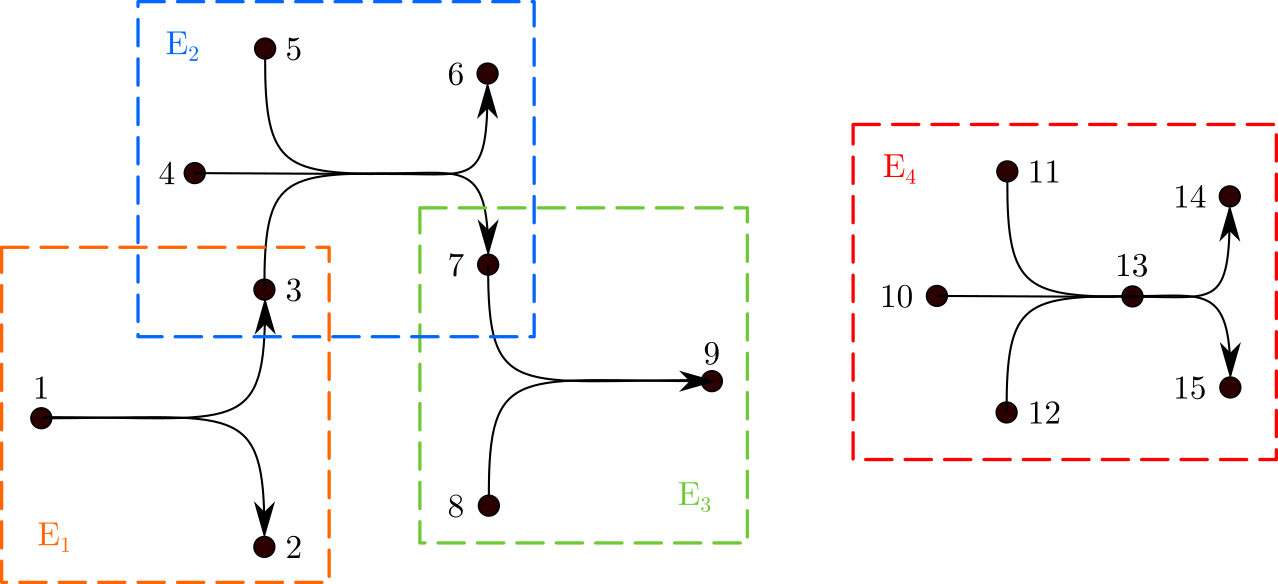}
    \caption{Example of a set-theoretic, directed hypergraph $H=(V,\{E_1,E_2,E_3,E_4\})$. The first three edges $E_1 = \{\{1\}, \{2,3\}\}, \, E_2 = \{\{3,4,5\}, \{6,7\}\}, \, E_3 = \{\{7,8\}, \{9\} \} $ are standard in directed hypergraphs. The last edge $E_4 = \{\{10,11,12\}, \{13\} , \{14,15\}\}$ is also sensible from a heuristic and set-theoretic point of view, but it does not fit in the input-output scheme. And even though we could reduce it to two edges (one with $13$ on its head, one with it on its tail), that defeats the purpose of higher order interactions: the same could be said about any hyperedge, it can be reduced to its pairwise constituents.}
    \label{fig:set-directed-hypergraph}
\end{figure}

This situation is reminiscent of the division between simplicial complexes and hypergraphs: the former is a more restrictive type than the latter, however if we are interested in a problem where its topological features and tools (e.g. homology) can be exploited \cite{battiston2021physics}, then it is advantageous to discuss it on its own. Here we can see that if the directedness is present in a binary fashion (input-output), these structures are worth considering, but we should keep in mind that there are other possible ways of non-undirectedness in the interactions of a hypergraph.


From a spectral centrality point of view, the directed hypergraph perspective is rather limiting, as the fundamental object in its study is the adjacency tensor and their spectral properties. While we can work out a way to successfully describe the adjacency tensor of these directed hypergraphs (see the next section), we can get more interesting structures if we start from an algebraic point of view. 

What we mean by this is, we will consider the adjacency tensor to the the \textit{fundamental} object describing the structure of a directed, uniform hypergraph, and we will from then make sense of the topology of the hypergraph. And, in certain cases, this will also provide a set-theoretic description, but there will be cases where no set-theoretic description is available. Hence, in the case of uniform hypergraphs, the set-theoretic point of view is a subset of a larger class of hypergraphs, which we call ``heterogeneous''. And even though we arrived at them via spectral centralities, it is worth noting that they are more ubiquitous.




Having said this, the space of possible heterogeneous, $m$-uniform hypergraphs is too vast (an adjacency matrix representing a directed graph can have $N^2$ different entries, and analogously an adjacency tensor can have $N^m$ different entries), and its analysis is therefore not possible in all generality. 

We will soon discuss some specific types of directedness depending on their tensorial representations. In that sense, and bringing back the terminology introduced in the undirected case, we will no longer have tensor components (i.e. hyperedges) symmetric under all possible permutations $\sigma(i_1 i_2 ... i_m)$. We will then show how these types of directedness and provide natural definitions for their transposition and strong connectedness (if any).


\subsubsection{Applications and examples of heterogeneous interaction systems}

The possibility of studying systems with heterogeneous interactions in the abovementioned sense seems so far like a mathematical pastime, bearing no connection to any real system, however that is far from reality. In fact, the original paper on directed hypergraphs \cite{gallo1993directed} already describes some use cases, and the need to describe heterogeneous interactions was already recognized more than a decade ago \cite{klamt2009hypergraphs}.

Here we briefly review some of the applications described in both references and we add some more which we think are of relevance:
\begin{itemize}
            
    \item Biochemical reactions. Online, open databases such as \cite{KIDA-datasets} or \cite{BiGG-datasets} can be mined in order to construct different directed hypergraphs where edges represent chemical reactions or metabolical pathways between metabolites, respectively \cite{klamt2009hypergraphs}. We performed this mining task, and the former example is later showcased in Subsection \ref{subsec:acyclic-BF}. 

    \item Citation networks. A classic example of networks and hypergraphs is that of authors as nodes, linked together if they are coauthors in a paper. In the pairwise case the coauthorship is established between any two authors if they wrote a paper together, possibly with others. In the undirected hypergraph case the paper itself is a single hyperedge between all coauthors. But this modelling dismisses the order within the list of authors, a key piece of information in many fields. This ordering can be introduced in the heterogeneous hyperedge case, but not in the directed case. 
    
    \item Urban transit. In \cite{gallo1993directed} it is argued that there are many mobility networks which can be abstracted as directed hypergraphs, where some nodes are stops within transportation lines (e.g. metro), and thus connected via pairwise edges representing said lines, while other nodes serve as the actual transit stops, which individuals need to traverse by walking to change lines \cite{criado2010hyperstructures}. These are represented as directed hyperedges.

    \item Routing/delivery transportation. It is relatively common for transportation businesses to have their trucks repeat a predefined path (e.g. to supply the same stores periodically). This has traditionally been modelled by dividing such paths into their pairwise components, however with heterogeneous hypergraphs we could consider the path itself to constitute a heterogeneous hyperedge, and encode it in the adjacency matrix as an asymmetric tensor component.
    
    A more restrictive, though more tractable version of this is the case where all paths are of the same length $k$. This is then related to the $k$-step eigenvector centrality originally put forward in \cite{xu2023two} in the $k=2$, undirected case, which we will extend in Subsection \ref{subsec:acyclic-kstep}.

    

    \item Online social interactions. Hypergraphs can be constructed from online forum data, where nodes represent users and hyperedges represent threads where they participate. Different amounts of participation can be encoded in the adjacency tensors via asymmetric tensor components, leading to a heterogeneous hypergraph.
    
    \item Propositional logic and relational databases. In \cite{gallo1993directed} different such hypergraphs are considered, with the most prominent type abstracting relational databases, where nodes represent propositions, with directed hyperedges linking a series of causes to their consequences. Notice that here the projection of the hypergraph to its pairwise constituents is losing key information (one needs all propositions to be fulfilled at once in order to prove their consequence).

\end{itemize}

It should be noted that we will not attempt to construct all of these hypergraphs, as this manuscript is of a theoretical nature, and we moreover are interested in those where spectral centrality measures can be computed (for instance, in propositional hypergraphs the avoidance of circular arguments renders eigenvector-like centralities impossible due to the lack of strong-connectivity). Nevertheless we will exemplify our methods with some real hypergraphs in the next section.

\subsection{The transposition operation} \label{subsec:prelim-transposition}

Before moving on to discuss spectral centralities in special heterogeneous hypergraphs, there is a technical detail which must be dealt with in all those cases.

Recall the standard Eigenvector Centrality (EC) of a graph. There, the transposition step is quite natural, either from a topological point of view (inverting the orientation of each edge), from a set-theoretic point of view (directed edges are ordered sets, transposition means inverting such order), or from an algebraic point of view (the matrix transposition maps components of a matrix into those of the transposed matrix). 

For this reason, computing the EC of a graph is equivalent to computing the Perron eigenvector of the ``topologically transposed'' or ``set-theoretical transposed'' graph. In some sense, there is, as we mentioned before, a clear connection between topology, set-theory and algebra. In the case at hand, said the connection is ``broken'', or rather, it is non-unique. Therefore, in order to make sense of the transposition, we need to establish ``by hand'' a clear connection between them. 


Our approach is an algebraic one, something which is quite reasonable taking into account that the $\cH$-eigenproblem is algebraic in nature. With that in mind, we can think of the transposition as a map from the space of tensor components to itself 
\begin{definition}[Transposition of an adjacency tensor]
    Let $\cT\in\R^{[m,N]}$ be an adjacency tensor and $\sigma$ a permutation of $m$ elements. The transposition of the tensor is defined, component by component, as $(T_{i_1 ... i_m})^t = T_{\sigma(i_1 ... i_m)},$
    where $\sigma(i_1 ... i_m)$ is the permutation of the indices $\{i_1, ..., i_m\}$. 
\end{definition}

Essentially, it is a choice of permutation of the indices involved, which in the matricial case is unique, i.e. $(A_{ij})^t = A_{ji}$. In the tensorial case this choice is not unique, but we can make it such that it is mathematically sensible for our problem, depending on the type of directedness we are considering.

Knowing what the transposition looks like from an algebraic point of view, we can now connect it to topology/set-theory.
\begin{definition}[Transposed hypergraph] \label{def:transposed-hypergraph}
    The transposed hypergraph $H^t$ associated to the hypergraph $H$ is the hypergraph induced by the transposed adjacency tensor $(\mathcal{T})^t$ of $H$.
\end{definition}

What the final structure (the transposed hypergraph) will look like can't be generally assessed, instead we will need to consider appropriate transpositions for each directed case, and that will force upon us the topological/set-theoretic nature of the transposed hypergraph.

What does this entail for the already established, undirected case? Our algebraic definition of a transposition in terms of index shifts is in complete agreement with what's been done so far for undirected hypergraphs. There, the transposition step is ignored altogether. And indeed, in an undirected hypergraph we have the exceptional property that $T_{i_1 ... i_m}=T_{\sigma(i_1 ... i_m)}\, \forall \sigma$, hence no matter which transposition rule we choose the outcome would be the same.

\section{Spectral centralities of heterogeneous hypergraphs} \label{sec:directed}

The study of spectral centralities in centrality in directed hypergraphs poses, as discussed above, a conceptual problem: What kind of transpositions are more mathematically coherent? 

It is not possible to answer it in all generality, for all possible, unconstrained tensors $\cT$. For that reason, we will be giving sensible definitions to the transposition of specifically structured tensors. These will represent either newly defined ones (cyclical hypergraphs and $k$-step hypergraphs) or more well-known ones (directed hypergraphs).

Also, depending on the constraints we place on the available hyperedges we will be able to give more specific details about the connectivity of the heterogeneous hypergraph. In that regard, it is important to keep in mind that the strong connectivity requirement for the existence and uniqueness of HEC is not a requirement on the hypergraph $H$, but on the transposed hypergraph $H^t$.

\subsection{Cyclical hypergraphs} \label{subsec:cyclic-hyperedges}

In this particular kind of heterogeneous hypergraph we restrict hyperedges to have a ``cyclicity'' permutation symmetry, as defined below.
\begin{definition}[Cyclical hyperedge]
    Let $\sigma_{\rm odd}$ and $\sigma_{\rm even}$ be the set of all even and odd permutations of $m$ elements, respectively. A cyclical hyperedge between nodes $\{i_1, \dots, i_m\}$ corresponds to 
    \begin{equation}
        T_{\sigma(i_1, \dots, i_m)} = T_{\sigma'(i_1, \dots, i_m)},\; \forall \sigma, \sigma' \in \sigma_{\rm odd} \quad \text{and} \quad T_{\sigma(i_1, \dots, i_m)} = T_{\sigma'(i_1, \dots, i_m)},\; \forall \sigma, \sigma' \in \sigma_{\rm even}
    \end{equation}
\end{definition}

For example, $T_{123} = T_{312} = T_{231}$ and $T_{132} = T_{213} = T_{321}$ correspond to the two possible cyclical hyperedges between nodes $\{1,2,3\}$.
Note that this type of heterogeneous hyperedges are already not possible to capture with a set-theoretical description (there is no split heads/tails).

A sensible definition for the transposition of the tensor in this scenario is mapping even permutations to odd permutations and viceversa by totally reversing the order of the indices
\begin{equation}
    (T_{i_1 \dots i_{m}})^t = T_{i_{m} \dots i_1}.
\end{equation}

In this type of hypergraph, the ``tensor apply'' operation \eqref{def:tensor-apply} required for centrality calculations reads
\begin{equation}\label{eq:tensor-apply-cyclic}
    \left[ (\cT)^t \bc^{m-1} \right]_{i_m} = \sum_{i_1, \dots,i_{m-1}=1}^n T_{i_{m} \dots i_1} \, c_{i_1} \dots c_{i_{m-1}}.
\end{equation}

Let us now address strong connectivity in regards to the transposition operation. In this case, a cyclic hyperedge in $H$ translates into a directed cycle $C_{m}$ in $G^M$ (see Definition \ref{def:strong-conn}), therefore we do not have to worry about the transposition in order to examine the connectivity, as either orientation of a directed cycle is strongly connected. 

In principle one can always study the HEC of a heterogeneous, $k$-uniform hypergraph, provided it is strongly connected, however this will yield no new results in comparison with its undirected counterpart, where $T_{\sigma(i_1...i_m)} = T_{\sigma_{even}(i_1...i_m)} + T_{\sigma_{odd}(i_1...i_m)}$. The reason for this is the fact that  \eqref{eq:tensor-apply-cyclic} does not distinguish between the cycle orientations, hence it yields uninteresting results.

While this particular kind of heterogeneous hypergraph is not interesting to study from a spectral centrality point of view, its existence and possible applications to model systems have yet to be unveiled.



\subsection{Directed hyperedges} \label{subsec:acyclic-BF}

We will now consider a directed hypergraph. This is still too broad to tackle, we need to narrow it down further. But before doing so, let us introduce some definitions from set-theoretic hypergraphs.

\begin{definition}[$\mathcal{B}$- and $\mathcal{F}$-hyperedges \cite{gallo1993directed}]
    A backward hyperedge, or simply $\mathcal{B}$-hyperedge, is a hyperedge $e = (t(e), h(e))$ with $|h(e)|=1$. A forward hyperedge, or simply $\mathcal{F}$-hyperedge, is a hyperedge $e = (t(e), h(e))$ where $|t(e)|=1$.
\end{definition}

It is common to refer to $h(e)$ as the ``head'' of the hyperedge, and to $t(e)$ as its ``tail'' \cite{gallo1993directed}. Notice that any hyperedge $e=(t(e), h(e))$ not belonging to either of them (i.e. $|t(e)|,|h(e)|\neq 1$) can be converted into a $\cB$-hyperedge and a $\cF$-hyperedge by splitting it placing a node between them (although from the point of view of centrality this is not an appropriate operation, for it creates new nodes with their own centrality). A hypergraph consisting of just $\cB$-hyperedges ($\cF$-hyperedges) is called a $\cB$-hypergraph ($\cF$-hypergraph).


Directed hypergraphs can be encoded as a series of adjacency tensors whose components satisfy certain symmetry constrains (this has already been proposed by \cite{gallo2022synchronization}). Each directed hyperedge contributes to the corresponding adjacency tensor as follows:
\begin{itemize}
    \item A backward hyperedge $E = (\{i_1 \dots i_{m-1}\}, \{j\})$ corresponds to tensor components 
    \begin{equation}
        T^{(m)}_{\sigma(i_1 \dots i_{m-1}) j},
    \end{equation} 
    for any permutation $\sigma$.
    \item A forward hyperedge $E = (\{i\}, \{j_2 \dots j_{m}\})$ corresponds to tensor components 
    \begin{equation}
    T^{(m)}_{i \sigma(j_1 \dots j_{m})},    
    \end{equation}
    for any permutation $\sigma$.
    \item A general hyperedge $E = (\{i_1 \dots i_{m'-1}\}, \{j_{m'} \dots j_{m}\})$ corresponds to tensor components 
    \begin{equation}
        T^{(m)}_{\sigma_1(i_1 \dots i_{m'-1}) \sigma_2(j_{m'} \dots j_{m})},
    \end{equation}
    for any permutations $\sigma_1,\sigma_2$.
\end{itemize}

{From an algebraic point of view}, the transposition of such a tensor corresponds to the index transposition 
\begin{equation}
(T^{(m)}_{i_1 \dots i_{m'-1} j_{m'} \dots j_{m}})^t = T^{(m)}_{j_{m'}  \dots j_{m} i_1 \dots i_{m'-1}},
\end{equation}
for each tensor $T^{(2)}, T^{(3)}, ..., T^{(M)}$. The following result is then straightforward.


\begin{proposition}[Transposed directed hypergraph]\label{prop:transposed-hypergraph}
The transposed hypergraph $H^t$ of a directed hypergraph $H=(V,E)$ is the result of interchanging $h(e)$ and $t(e)$ in each hyperedge $e\in E$.
\end{proposition}

As we will see, the distinction between the two types of hyperedges will have dramatic consequences in centrality computations.

\subsubsection{$\cB$-hypergraphs}

The first case of directed hypergraphs we will consider is that where only $\cB$-hyperedges are present, with a fixed number of nodes per hyperedge (i.e. $m$-uniform). In this case the transposition reduces to
\begin{equation}
(T_{i_1 \dots i_{m-1} j})^t = T_{j i_1 \dots i_{m-1}}.
\end{equation}

This makes it perfectly suitable for the HEC centrality calculations, as the free index in the sum is the first one, and the sums run equally over the remaining indices. It is important to notice that the $\mathcal{H}$-eigenvector one needs to compute is that corresponding to the transposed hypergraph, which will be the one whose connectivity is called into question. 

The ``tensor apply'' operation in this case becomes
\begin{equation}
    \left[ (\cT)^t \bc^{m-1} \right]_j = \sum_{i_1, \dots,i_{m-1}=1}^n T_{j i_1 \dots i_{m-1}} \, c_{i_1} \dots c_{i_{m-1}}.
\end{equation}

However, for this operation to make sense, we need to impose uniformity in the number of head nodes of a hyperedge, so as to be able to construct a single tensor from them. Hence we arrive at the following definition.

\begin{definition}[$\cB$-uniformity]
Let  $H=(V,E)$ be a directed hypergraph, $m\in \{1,\dots, \max_{e\in E}{|e|}\}$. $H$ is said to be $\cB$-uniform if $m=|h(e)|, \forall e \in E$. 
\end{definition}

We will discuss and compute numerically centralities of hypergraphs fulfilling this constraint in a later subsection, but before that let us examine the other kind of hyperedges.


\subsubsection{ $\mathcal{F}$-hypergraphs}

We now turn to the opposite case, that of directed hypergraphs containing only $\cF$-hyperedges, with a fixed number of nodes per hyperedge (i.e. $m$-uniform)\footnote{
This type is employed in \cite{xie2016spectral, qi2017tensor} with a specifically weighted adjacency tensor 
}. As we will see, this case is rather troublesome when it comes to centrality computations.

In this case the transposition reduces to
\begin{equation}
(T_{i j_2 \dots j_m})^t = T_{j_2 \dots j_m i}.
\end{equation} 

If we naïvely adapt the ``tensor apply'' operation, we get a mixing of the indices summed over due to the transposition, in the sense that we sum over both the tail node and some head nodes. 
\begin{equation}
    \left[ (\cT)^t \bc^{m-1} \right]_{j_2} = \sum_{j_3,  \dots, j_{m},i=1}^n  T_{j_2 \dots j_m i} \, c_{j_3} \dots c_{j_m} c_{i}.
\end{equation}

This is heuristically incorrect: in directed edges the centrality contribution flows from the ``input'' node(s) to the ``output'' one(s), a fact that is later reflected in the sum. However, if we take this operation at face value, we see that all head nodes will contribute to each other. 

We argue that the correct operation involved in the centrality calculation is not the standard ``tensor apply''. Instead, centrality in $\cF$-hypergraphs should be calculated using
\begin{equation}
     \lambda c_{j} = \sum_{i\rightarrow \{j, \dots, j_m\}} c_i = \frac{1}{(m-1)!} \sum_{j_3, \dots, j_{m}, i = 1}^n  T_{j j_3 \dots j_m i} \, c_i,
\end{equation}
which is, in essence, a glorified version of the eigenvector centrality (as indeed there are $(m-1)!$ identically valued components $T_{j\sigma(j_3 \dots j_m) i}$). 

This implies that the tensorial nature of the hypergraph is actually shadowed in this case by its pairwise directed relations. Nevertheless, there is at least a unique and easy to compute solution to this problem. And this further implies that our initial restriction to $m$-uniform $\cF$-hypergraphs was not necessary: we can include hyperedges of any given size in the sum above, provided we take into account the $1/(m-1)!$ factors, with $m$ the size of the corresponding $\cF$-hyperedge.

\subsubsection{General directed hypergraphs}

Now that we understand how does the ``tensor apply'' operation act on both forward and backward hyperedges, we can discuss general directed hypergraphs and their centrality. 

First, note that backward hyperedges impose a restriction on the number of nodes at the input of every interaction, namely we will require the hypergraph to be $\cB$-uniform. On the contrary, forward hyperedges impose no restriction, as their tensor apply operation needs to be replaced by a projection-like operation. 

Therefore, for a $\cB$-uniform directed hypergraph we can define the following generalizations of the HEC centrality measures, provided it is strongly connected.

\begin{definition}[HEC of a directed hypergraph]
Let $H=(V,E)$ be the a strongly connected, $\cB$-uniform hypergraph, with $m_B=|h(e)|,\, \forall e\in E$ head nodes, and $m_F=\max(|t(e)|),\, e\in E$ the maximum number of tail nodes in any hyperedge. The $\cH$-eigenvector centrality of $H$ is the unique, positive vector $\bc\in \R^n$ satisfying
\begin{equation}
    \lambda c_{j_1} = \sum_{i_1, \dots,i_{m_B}=1}^n \, \sum^{m_F}_{m=1} \frac{1}{m!}  \sum_{j_1, \dots, j_{m}}^n  T_{j_1 \dots j_m  i_1 \dots i_{m_B}} \, c_{i_1} \dots c_{i_{m_B}}.
\end{equation}
\end{definition}

Although this definition may seem daunting, operationally it is rather clear: Given a $\cB$-uniform hypergraph where $m_B$ is the number of head nodes (fixed due to the uniformity constraint), we turn every hyperedge into several ones, one per tail node, adjusting the combinatorial factors to avoid overcount due to the symmetry of the adjacency tensor. We are then left with a tensor $\cT\in\R^{[m_B,N]}$, whose centrality we can compute as its Perron-like eigenvector, like in the undirected case \cite{benson2019three}.

\subsubsection{Numerical examples}

Finding datasets of directed hypergraphs freely available is far from trivial, especially if we are interested in real data ones rather than synthetic ones. However, we managed to construct several of them related to chemical \cite{RTG-datasets} and astro-chemical reactions\footnote{We also constructed others from metabolical reactions \cite{BiGG-datasets}, but for the sake of conciseness the computation of the rankings and comparisons are left in the repository, see Data Availability section.} \cite{KIDA-datasets}. Once the hypergraphs have been constructed \cite{XGI-library}, we find the strongly connected sub-hypergraph satisfying the $\cB$-uniformity condition, if any.

From these hypergraphs we can then compute their directed $\cH$-eigenvector centrality (using the code we already developed and made available in \cite{contrerasaso2023uplifting}) and that of their projection graph. The top $5$ nodes in each example, as obtained via the standard eigenvector centrality of the projected graph, as well as with the newly defined measure for directed hypergraphs, are shown for comparison in Table \ref{tab:top5ranking}.

\begin{table}[h!]
    \centering
    \begin{tabular}{|c|cc|cc|cc|cc|cc|}
    \toprule
        - & \multicolumn{2}{c}{\bf RTG}  & \multicolumn{2}{c}{\bf unibimolecular} & \multicolumn{2}{c}{\bf surface} & \multicolumn{2}{c}{\bf astro} \\ 
        \midrule
        - & {\bf EC} & {\bf HEC} & {\bf EC} & {\bf HEC} & {\bf EC} & {\bf HEC} & {\bf EC} & {\bf HEC} \\ 
        \midrule
        1 & H & OH & H & e- & H2 & H & H & E \\ 
        2 & OH & H & H2 & Photon & CH4 & CH3SH & H2 & C \\ 
        3 & O & O &  CO & O & H & CH3 & CO & C+ \\ 
        4 & CH3 & CH3 &  C & C & CH3 & H2CS & C & O \\ 
        5 & HCO & HCO &  He & C+ & NO & HNO & CN & H+ \\ 
    \bottomrule
    \end{tabular}
    \caption{First 5 elements ranked by both the Eigenvector Centrality (EC) of the projected graph and the directed $\cH$-eigenvector Centrality (HEC), in each of the datasets (RTG, unibimolecular, surface and astro).}
    \label{tab:top5ranking}
\end{table}

We can also compare these measures in terms of the Spearman's $\rho$ coefficient between both measures, which measures the similarity of two rankings (ranging from 1, meaning identical rankings, to -1, meaning completely opposite rankings). The correlation between rankings for the first $K$ components of each hypergraph are shown in Figure \ref{fig:chemical-correlation}.  {Notice that, even though the covariance is symmetric, the rankings H3-E and E-H3 differ: that is due to the fact that, before the last datapoint, the correlation compares the first $K$ nodes in each measure with their relative position in the other, but these first $K$ nodes need not be the same on each measure.}

\begin{figure}[h!]
    \hspace{-0.4in}
    \includegraphics[width=1.1\textwidth]{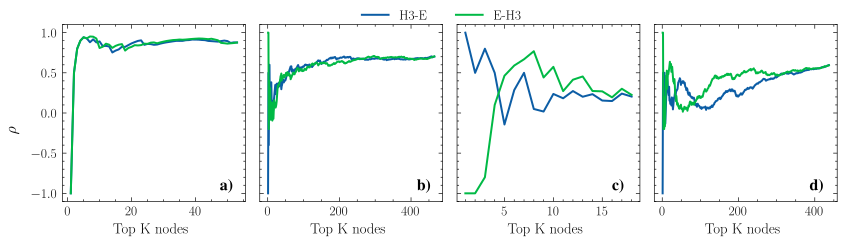}
    \caption{Comparison between Spearman's $\rho$ correlation between the Eigenvector Centrality (E) of the projected graph and the directed $\cH$-eigenvector centrality (H3), in the following real chemical reaction dataset hypergraphs: a) RTG b) unibimolecular c) surface d) astro. }
    \label{fig:chemical-correlation}
\end{figure}

From these figures we can see that the rankings from the pairwise projection networks and the ones from our new method are positively correlated (as they should), being very similar in some cases (e.g. subfigures \textbf{a} and \textbf{b}) but different enough that they represent a different meaning of what importance is. {The ``oscillations'' in the intermediate range point to the fact that, even though the correlation tends stabilize when we consider most nodes, there are some intermediately ranked nodes in a measure, considered very relevant in the other, causing the down spikes. We have also briefly noted this in the manuscript.}


\subsection{$k$-step hypergraphs} \label{subsec:acyclic-kstep}

In \cite{xu2023two} the authors designed a new centrality measure (the \textit{two-step eigenvector centrality}) for undirected graphs, which consisted of constructing a tensor (the \textit{two-step tensor}) from powers of the adjacency matrix of the graph, which could then leverage the machinery of $\cH$-eigenvectors in order to rank the nodes of the original graph\footnote{Said article is only concerned with undirected graphs and $\cH$-eigenvectors, however there is actually no roadblock for extending it to directed graphs and $\cZ$-eigenvectors.}.

\begin{definition}[2-step centrality of a graph \cite{xu2023two}]
    Let $G=(V,E)$ be an undirected graph with $N$ nodes, let $A=(a_{ij})$ be its adjacency matrix. The 2-step eigenvector centrality of the graph is the Perron-like $\cH$-eigenvector solution to
    \begin{equation}
        \lambda c_i = \sum_{j,k=1}^N T_{ijk} c_{j}c_k, \quad \text{where} \quad T_{ijk} = a_{ij}a_{jk}. 
    \end{equation}
\end{definition}

Here we will study a more general version of this method: constructing the $(k-1)$-step tensor of a directed graph and analyzing its properties. 

\begin{definition}[$k$-step hyperedge]
Let $G=(V,E)$ a possibly directed graph with adjacency matrix $A=(a_{ij})$. Consider a sequence of its nodes $i_1\rightarrow i_2 \rightarrow \dots \rightarrow i_{k}$. A $k$-step hyperedge corresponding to such sequence is described by the adjacency component  
\begin{equation}
    T_{i_1 i_2 \dots i_{k}} = a_{i_1 i_2 } a_{i_2 i_3} \dots a_{i_{k-1} i_k}, 
\end{equation}
where $T_{i_1 i_2 \dots i_k}>0$ iff that walk exists in the network.
\end{definition}

It is not hard to see that a tensor constructed this way represents yet another type of heterogeneous hypergraph (one where order matters within a hyperedge but where there are no heads/tails), which we now analyze.

The most natural transposition in this case is that inherited from the transposed adjacency matrix, i.e.
\begin{equation}
(T_{i_1 i_2 \dots i_{k-1} i_k})^t = a^t_{i_1 i_2 } \dots a^t_{i_{k-1} i_k} = a_{i_k i_{k-1} } \dots a_{i_2 i_1} =  T_{i_k i_{k-1}\dots i_2 i_1}.
\end{equation} 

Clearly, for an undirected network we have the symmetry property $T_{i_1 i_2 \dots i_{k-1} i_k} = T_{i_k i_{k-1} \dots i_2 i_1}$, which is why in \cite{xu2023two} they safely ignore the transposition step, although omitting this discussion.

We have the following, completely standard formulation of the ``tensor apply'' operation.
\begin{equation}
    \left[ (\cT)^t \bc^{m-1} \right]_{i_k} = \sum_{i_{k-1}\dots, i_1=1}^n  T_{i_k i_{k-1}\dots i_1} \, c_{i_{k-1}} \dots c_{i_1}.
\end{equation}

In some sense this measures provides centrality to the target node from its $k-1$ ancestors. This is very reminiscent of centrality measures involving the iterated line graph or the Hashimoto matrix in standard networks \cite{krzakala2013spectral, aleja2019nonbacktracking}.


\medskip

Strong connectedness of the resulting tensor is guaranteed if the base graph is strongly connected by the following theorem.

\begin{theorem}[Strong connectedness of the hypergraph induced by $T_{i_1\dots i_k}$]
Let $G=(V,E)$ be an unweighted graph with adjacency matrix $A$. Let $H$ be the hypergraph induced by the tensor $T_{i_1 i_2\dots  i_m}=a_{i_1 i_2} a_{i_2 i_3} \dots  a_{i_{k-1} i_k}$. If $G$ is strongly connected, then $H$ is strongly connected.
\end{theorem}
\begin{proof}
    Consider the definition of strong connectivity based on the graph $G^M$ (see Definition \ref{def:strong-conn})
\begin{equation}
    M_{ij} = \sum_{i_3, \dots, i_k}^N T_{ij i_1 \dots  i_k} = \sum_{i_3, \dots,  i_k}^N a_{ij} a_{ji_3} \dots  a_{i_{k-1} i_k}=  a_{ij} \sum_{i_3, \dots,  i_k}^N a_{ji_3} \dots  a_{i_{k-1} i_k} \geq a_{ij},
\end{equation}
because if $G$ is strongly connected then there must be at least one $(k-2)$-long walk starting from $j$ (possibly repeating edges and nodes). Therefore $M \geq A$ component-wise, so if $G$ is strongly connected, $G^M$ is strongly connected and so is $H$. 
\end{proof}

Extending the proof for positively weighted graphs is straightforward (in which case we have $M_{ij} > 0$ instead). Note that this Theorem applies to both undirected as well as directed graphs.



\subsubsection{Numerical comparisons}

To illustrate the difference between the usual eigenvector centrality and this new measure, we will compute them for three real directed networks: the Chicago road network (12982 nodes, 39018 edges), the European road network (1174 nodes, 1417 edges) and the OpenFlights dataset (2939 nodes,	30501). The choice of transport networks is due to the intrinsic relation between the $k$-step hypergraphs and the routing transportation application which we discussed at the end of Section~\ref{sec:heterogeneous}. 

These results again showcase the disparity in rankings with the Spearman's $\rho$ correlation, which signals that the new measure is capturing different features (namely, it takes into account all path of length $k$) in the computation of the centrality scores. The results are shown in Figure \ref{fig:kstep-real}.
\begin{figure}[h!]
    \hspace{-0.4in}
    \includegraphics[width=1.1\textwidth]{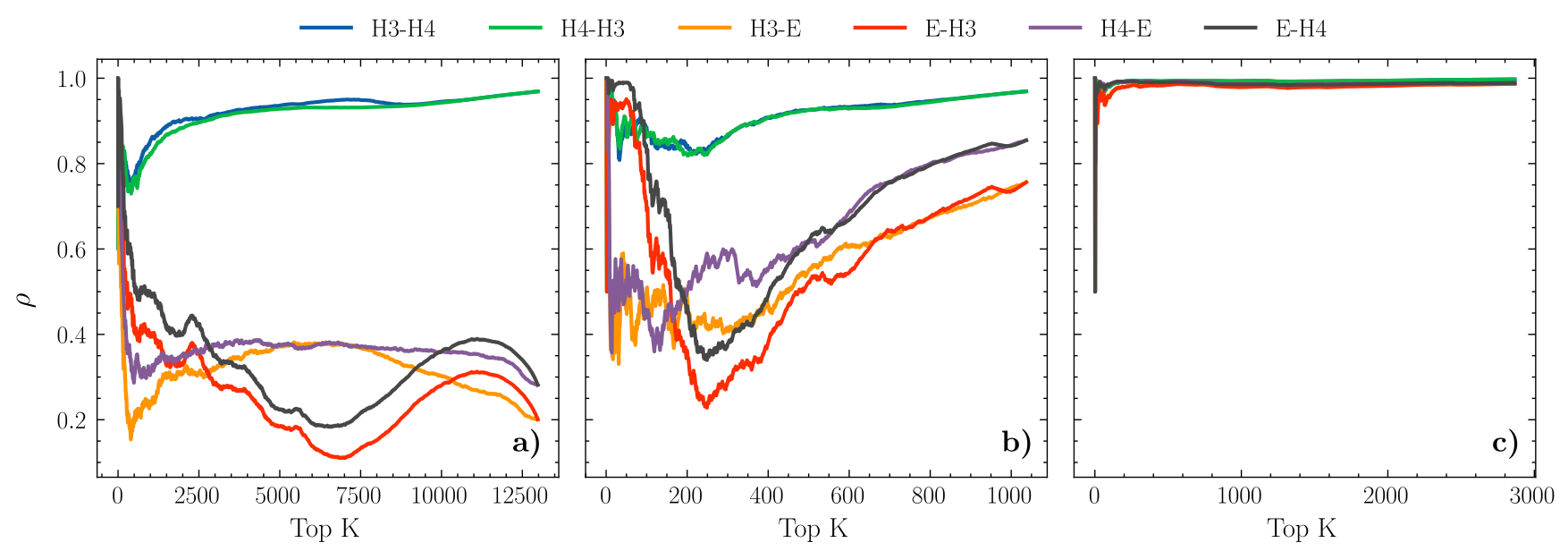}
    \caption{Comparison between Spearman's $\rho$ correlation between the Eigenvector Centrality (E) of the original directed graphs and their $k$-step eigenvector centrality for $k=3,4$ (H3, H4) in the case a) Chicago roads b) European roads c) OpenFlights}
    \label{fig:kstep-real}
\end{figure}

In the first two subfigures one can clearly see that the correlation between the 3-step and 4-step HEC of the graph are very similar, which means that there is almost no new information when increasing the length of paths sampled with the $k$-step centrality. The other comparisons show a higher disparity between rankings, meaning each measure is measuring something different, being positive nonetheless.

In the last subfigure, corresponding to a flight transportation network, we find it remarkable that the rankings are almost identical, having the Spearman correlation very close to 1 for all $K$ sampled. This last dataset corresponds to an undirected network, which, for the purposes of this article, is treated as a directed one with both directions per edge. However this bears no relation with the high correlation, and indeed we have checked that other undirected networks (e.g. the US Power Grid) do not have this feature.

\section{Conclusions} \label{sec:conclusions}

In this paper, we embarked on a journey through the realm of complex systems, where the traditional directed relationships of graph theory are transcended by the multifaceted interactions captured by hypergraphs. These structures, enriched by hyperedges linking multiple nodes, offer a more expressive framework for understanding complex systems across diverse disciplines.

A key concept we focused on was that of centrality, fundamental in network analysis as a mean to quantify the relative significance of nodes within a network. Spectral centralities have a special place for their analytical tractability and fast computation, but their adaptation to hypergraphs presents challenges which demand a mathematically rigorous response.


Our paper sheds light on the wide variety of non-undirected hyperedges, beyond traditional directedness, acknowledging diverse forms of multiple interactions among individuals. By introducing the term ``heterogeneous hyperedge'', we aimed to broaden the understanding of non-undirected interactions, including directed hyperedges as a specific type. {In order to convey the need of these structures in real data, we listed several real systems which would benefit from such a description; future studies and data collections will surely reveal the relevance in these scenarios. From a more theoretical point of view, several problems are now open for discussion: any prior works relying on directed hypergraphs (e.g. hypergraph curvature \cite{leal2021forman}, synchronization \cite{gallo2022synchronization} or community detection \cite{contreras2023detecting}) need to be extended to heterogeneous hypergraphs, specially when adjacency representations are involved.}

We not only raised awareness about the existence and significance of heterogeneous interactions but also proposed extensions of spectral centralities to various forms of heterogeneous hyperedges. Through numerical simulations on real and synthetic hypergraphs, we illustrated the practical implications of these extensions.

As we conclude our study, we recognize the complexity and richness inherent in the study of directed hypergraphs and heterogeneous interactions. By laying the groundwork for further exploration, we hope to inspire future research endeavors aimed at unraveling the intricate fabric of complex systems' dynamics.


\subsection*{Funding}

This work has been partially supported by projects M3033 (URJC Grant) and the INCIBE/URJC Agreement M3386/2024/0031/001 within the framework of the Recovery, Transformation and Resilience Plan funds of the European Union (Next Generation EU). G. C-A. is funded by the URJC fellowship PREDOC-21-026-2164.

\subsection*{Data Availability Statement}

The data used for the numerical results presented in this work, as well as the code written to analyze and plot them, can be found in the repository 

\url{https://github.com/LaComarca-Lab/DirectedHyperCentrality}.





\printbibliography



\end{document}